%% file: Eternal.tex
\documentclass[aps, prl, twocolumn]{revtex4}

\input{preamble.tex}

\begin{document}

\title{The time-reverse of any causal theory is eternal noise}
\author{Bob Coecke and Stefano Gogioso}
\affiliation{University of Oxford, Department of Computer Science.}
\author{John H.~Selby}
\affiliation{Imperial College, Department of Physics,} 
\affiliation{University of Oxford, Department of Computer Science}
\affiliation{Perimeter Institute for Theoretical Physics.}

\begin{abstract}
\noindent We consider a very general class of theories, \emph{process theories}, which capture the underlying structure common to  most theories of physics as we understand them today (be they established, toy or speculative theories). Amongst these theories, we 
focus on those which are \emph{causal}, in the sense that they are intrinsically compatible with the causal structure of space-time as required by relativity. We demonstrate that there is a sharp contrast between these theories and the corresponding time-reversed theories, in which time is taken to flow backwards from the future to the past. While the former typically feature a rich gamut of allowed states, the latter only allow for a single state: eternal noise. We illustrate this result by considering the time-reverse of quantum theory.  Moreover, we 
derive a strengthening of the result in PRL \textbf{108}, 200403 on signalling in time-reversed theories.
\end{abstract}


\maketitle


\noindent The question of whether theories of physics should be formulated in a time-symmetric fashion has received a lot of attention in the past few years \cite{aharonov2010time, CoeckeLal, d'ariano2014determinism,oreshkov2015operational,oreshkov2016operational,selby2017diagrammatic}, because of its deep foundational implications and the pivotal role it plays in the relationship between quantum theory and general relativity. In this work, we show 
that if a theory obeys the
requirements imposed by
relativistic principles then the time-reverse of the theory is necessarily reduced to a trivial state of eternal noise. This paints a provocative picture, in which time-symmetry breaks down in the most extreme manner imaginable, but is the natural consequence of taking general relativity seriously and applying its constraints to the letter.

The assumption that nature is time-symmetric arises from taking reversible dynamics as the primary object of importance for a theory. In contrast, we choose to work within the framework of \emph{process theories} \cite{CKpaperI, CKBook}, which forces one to treat all processes -- including states, effects and general transformations -- on an equal footing with reversible dynamics.  Our main result is then an immediate consequence of compatibility of processes with relativistic constraints. 


\section{Process Theories}

\noindent A process theory  is comprised of 
a collection of \emph{systems}, graphically represented by wires, and \emph{processes}, graphically represented by boxes with wires as their inputs (at the bottom) and outputs (at the top). Special cases of processes include: \emph{states}, which are processes with no inputs; \emph{effects}, which are processes with no outputs; \emph{numbers}, which are processes with neither input nor output.
Processes can be
wired together in order to form
\em diagrams \em (a.k.a.~\em circuits\em) such as in (\ref{eq:compound-process}) below. 
Note that, further to assigning a meaning to wires and boxes, a process theory should explicitly specify what it means to form diagrams; in particular, every diagram must itself correspond to a process in the theory.
In the  case of diagram (\ref{eq:compound-process}),  this is a process with input $A$ and output $ABC$, i.e.~a \emph{composite} system.
In algebraic terms, process theories are  
\emph{(strict) symmetric monoidal categories} \cite{CatsII}.

Recall that we adopted a convention whereby process inputs are drawn at the bottom, and process outputs are drawn at the top. When drawing diagrams, this means that the upward vertical direction is associated with time flowing forward:

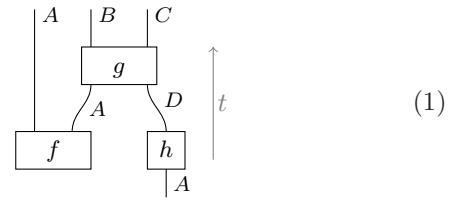
\beq\label{eq:compound-process}
\input{compound-process-time.tikz}
\eeq


\noindent The rather general formulation of process theories has helped to solve a number of concrete problems in quantum foundations and quantum computation. In quantum foundations, for example, it has provided a complete and consistent formulation of Spekkens' toy theory, and pinpointed its key differences from quantum theory \cite{CES, MiriamSpek}. In quantum computation, it has resulted in the derivation of exact correspondences between quantum computational models \cite{DP2, BoixoHeunen, Clare} and error-correction \cite{Chancellor2016Coherent-Parity, BH-2017}.

One important thing to remark about process theories is that the diagrams capture \emph{all} there is to know about the theory. This includes all relevant systems are represented explicitly in the diagrams e.g.~also those carrying information about measurement outcomes and classical control in the case of quantum theory. This is in contrast to other `hybrid' formalisms, such as \cite{AC1,HardyJTF,Chiri1}, where classical systems are represented implicitly by means of indices and labels.

One example of a process theory is 
probability theory (a.k.a.~classical theory), with probability distributions as states and stochastic maps as processes.
Another interesting example is given by quantum theory: processes include classical processes between classical 
systems, 
quantum processes between quantum systems (i.e.~CPTP maps), as well as mixed classical-quantum processes (e.g.~preparations, measurements, 
instruments, controlled unitaries, etc.) \cite{CKBook}.
Many other physical theories of interest can be formulated as process theories, even though they traditionally might not be. General relativity and quantum field theory, in particular, tend to take a ``block-universe'' approach (following Minkowski), as opposed to the ``compositional'' approach of process theories. However, recent work \cite{hardy2016operational,gogiosogenovese2017qft} has provided initial evidence that those theories can be reformulated in a suitably compositional way, and be turned into process theories.



\section{Process Terminality}


\noindent Our starting point is the following principle, underpinning general scientific practice: the description of reality here and now \emph{cannot} depend on things happening too far away in space, or in the future. In general relativity, this principle manifests itself as
the causal structure of space-time.
From the point of view of process  theories, this means that there should be a way of
discarding  those systems which  don't affect our local description of reality, compatibly with the constraints imposed by relativity. 
Implementing this prescription can be achieved by
providing a designated \emph{discarding effect} for each system in the theory:
\[
\input{disc.tikz}
\]
A special case of the discarding effect appears when there is nothing to actually discard
e.g.~at the output of an effect. In that case, the discarding effect is simply denoted by an empty diagram:
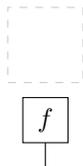
\[
\input{disc-effect.tikz}
\]
In order for discarding effects to accurately model the idea that a system
has been discarded,
the processes in the theory have to obey the following
constraint:
discarding all outputs after a process $f$ has occurred is equivalent to discarding all inputs beforehand; either way the system ends up discarded. Formally, the constraint is expressed by the following definition.


\begin{definition}\label{def:terminal}
A process $f$ is \emph{terminal} if it satisfies:
\begin{equation}\label{eq:term}
  \input{term.tikz}
\end{equation}
and  a process theory is \emph{terminal} if all its processes are.
\end{definition}

Note that from this formal definition it clearly follows that we mean `truly discarding', i.e.~the process did not ``stash away'' any of its input, as there are no outputs, and diagrams are all there is to know about the theory.   Within the more restricted context of  Operational Probabilistic Theories \cite{Chiri1},  Definition \ref{def:terminal} corresponds to the principle of `no signalling from the future', which was introduced as a postulate in
\cite{Chiri1}
and shown to be equivalent to
uniqueness of  deterministic effects.

The special of equation (\ref{eq:term}) where $f$ is an effect immediately implies that the only effects are the discarding ones themselves:
\beq\label{uniqueEffects}
\input{uniquenessOfEffects.tikz}
\eeq

\noindent In the case of  probability theory, discarding effects are row vectors with all entries equal to 1, and in particular the
empty diagram is the number 1 itself. Discarding a state in classical probability theory (a column vector with non-negative real components) amounts to producing a number equal to the sum of its components: hence a terminal state is one with components summing to 1, i.e.~a probability distribution. More generally, discarding a process (a matrix with non-negative real components) amounts to producing a row vector equal to the sum of its rows: hence a terminal process is a matrix with entries of each column summing up to 1, i.e.~a stochastic matrix. In the case of quantum theory, the discarding effect for a quantum system is the process of taking the (partial) trace
 of the system: hence terminal states are unit-trace density matrices, and terminal processes are trace-preserving CP maps.

It has been shown \cite{Chiri1,coecke2016terminality} that equation (\ref{eq:term}) allows one to derive the ``no-signalling'' principle, showcasing  a close connection between the notion of terminality and relativity theory. A much stronger result was obtained in \cite{IFF}, where it was shown that the very existence of relativistic causal structure is equivalent to terminality of a process theory.

\begin{theorem}[\cite{IFF}]
A process theory is \emph{causal}---in the sense of general relativity---if and only if it is \emph{terminal}.
\end{theorem}

\noindent The importance of this result to the
understanding of physics stems from the the different mathematical standing of causal structure and process terminality. Traditionally, causal structure is defined by a partial ordering between events \cite{kronheimer1967structure,MartinPanagaden,Malament}, with the interpretation that event $a$ can affect event $b$ if and only if $a \leq b$. Under this definition, causal structure is not \emph{intrinsic} to a process theory, instead having to be imposed \emph{on top} of it. By passing 
to process terminality, relativistic constraints are directly \emph{incorporated} into the process theory, which becomes a full-fledged, stand-alone description of physical reality.

\section{Main result}

\noindent
A major advantage of the process-theoretic format is the canonical manner in which time-reversal arises, even in absence of a time parameter $t$:
it is 
simply obtained by inverting the temporal reading order of the diagrams, turning our bottom-to-top convention into a top-to-bottom one without altering the diagrams themselves:
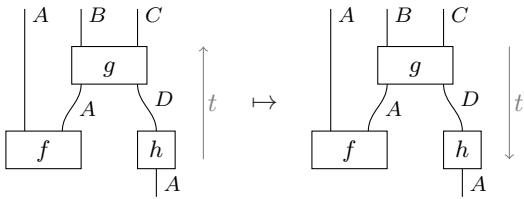
\[
\input{compound-process-time.tikz}\quad\mapsto\quad\input{compound-process-time-rev.tikz}
\]
Alternatively, we can retain our bottom-to-top reading convention, and flip the diagrams themselves upside-down, obtaining the corresponding \emph{time-reversed diagrams}:
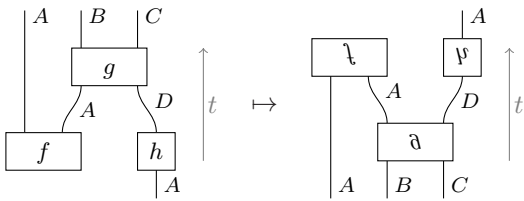
\[
\input{compound-process-time.tikz}\quad\mapsto\quad\input{compound-process-time-flip.tikz}
\]
Notably, our definition of time-reversal does not include a specification of what time-reversing a box actually means: the reason for this is that such specifics are of no importance to our main result, which applies to any notion of time-reversal which respects the structure of diagrams.
In the special case where the original process theory is causal, we refer to its time-reverse as a \emph{retro-causal} process theory.
\begin{theorem}[Eternal noise]\label{thm-time-rev}
In a retro-causal process theory every system has a unique state, namely the time-reversal of the discarding effect:
	\begin{equation}\label{eq:disc-rev}
		\input{disc-rev.tikz}
	\end{equation}
\end{theorem}
\begin{proof}
By flipping the process terminality condition of equation (\ref{eq:term}) upside-down we obtain the following equation, holding for every process $\timerev{f}$ in the time-reverse of the theory:
\[
  \input{term-rev.tikz}
\]
In the special case where the process $\timerev{f}$ is chosen to be a state, we have that the process $f$ in the original theory must necessarily be an effect. But $f$ has to be terminal, i.e.~it must be the discarding effect itself, and hence $\timerev{f}$ is always necessarily equal to the time-reverse of the discarding effect (we have used the fact that the time-reversal of the empty diagram is again the empty diagram). 
This concludes our proof.
\end{proof}

\noindent A system with a single state cannot possibly carry any information, so we can naturally interpret this state as \em (white) noise\em. Taking the standard notion of time reversal within quantum and classical theory (as provided by the Hermitian adjoint and transpose respectively \cite{CKBook}), we find that this is explicitly the case. Hence the result of Theorem \ref{thm-time-rev} can be restated as follows: any retro-causal process theory is necessarily in a state of \emph{eternal (white) noise}.

\section{Example: Time-reversing quantum theory}

\noindent To illustrate the extremely general result from the previous section, we will now consider the concrete example of quantum theory. Quantum theory is  captured by a process theory in which systems correspond to finite dimensional Hilbert spaces, and processes to CPTP maps between them. Typically time-reversal is described by the Hermitian Adjoint, which  respects the structure of diagram as required by our assumptions  \cite{CKpaperI, CKBook}.

If we consider reversible dynamics for a quantum system, i.e.~a unitary $U$, then the  corresponding  retro-causal process is $\timerev{U}=U^\dagger$.  Thise is clearly a suitable choice of time-reversed process, since it `undoes' the evolution of the quantum system as expected:
\[
\input{unitary.tikz}
\]
Note that any unitary will belong to both the causal and retro-causal theory:
this follows from the fact that the time-reverse of the discarding effect is
proportional to the maximally mixed state, 
and it means that unitaries  are \emph{unital} CPTP maps. Because unitaries belong to both, we must go beyond reversible dynamics to observe the distinction between the causal and retro-causal theories. Specifically, we will consider measurements, i.e.~
physical interactions between quantum and classical systems. In order to readily distinguish between classical and quantum systems, we will adopt the graphical convention of \cite{CKBook}, by which quantum systems are denoted by thick wires and classical systems are denoted by thin wires.
For example, a non-demolition measurement associated with a quantum observable---a complete family $(P_x)_{x \in X}$ of orthogonal projectors on a quantum system---is a process with quantum input and quantum-classical output: it takes a quantum state, non-deterministically applies a projector $P_x$ form the observable to it, and returns the projected state in output, together with the classical label $x \in X$ identifying the projector that has been applied. In the diagrammatic style of \cite{CQMII, CKBook}, one such non-demolition measurement takes the following form:
\begin{equation}\label{eq:meas-explicit}
\input{meas-explicit.tikz}
\end{equation}
The corresponding demolition measurement is the process which is obtained by discarding the quantum output:
\begin{equation}\label{eq:dem-meas-explicit}
\input{dem-meas-explicit.tikz}
\end{equation}
The demolition measurement is a quantum-to-classical process. Applied to a density matrix $\rho$ it results in a probability distribution $\mathbb{P}[x \in X]$ on the set $X$ indexing the projectors:
\[
\mathbb{P}[x \in X] = \operatorname{Tr}\left(P_x \rho\right)
\]
The time-reverse of the non-demolition measurement from equation \ref{eq:meas-explicit} is the following process:
\begin{equation}\label{eq:meas-rev-explicit}
\input{meas-rev-explicit.tikz}
\end{equation}
The classical outcome is now classical control, and picks out a choice of projector to be applied to the quantum system. Similarly, the time-reverse of the demolition measurement from equation (\ref{eq:dem-meas-explicit}) is now a classically controlled preparation of quantum states:
\begin{equation}\label{eq:dem-meas-rev-explicit}
\input{dem-meas-rev-explicit.tikz}
\end{equation}
In a retro-causal theory, however, the only classical state available to control these projector-selection and state-preparation processes is the uniform probability distribution over controlling states. As a consequence, all we can achieve by classical control in time-reversed quantum theory is decoherence in an observable and preparation of the totally mixed state:
\begin{equation}\label{eq:meas-rev-explicit-uniform}
\input{meas-rev-explicit-uniform.tikz} \hspace{2cm} \input{dem-meas-rev-explicit-uniform.tikz}
\end{equation}

\noindent This
observation resolves an apparent contradiction between this work and  the  result in \cite{CoeckeLal} stating that the time-reverse of classical probability theory is a signalling theory:
in principle, the retro-causal classical theory \emph{does} admit signalling processes, but there is no information to actually signal, only eternal noise.  One might therefore be tempted to ask: what happens if we de-couple classical agents from the theory that we reverse, and retain their normal abilities?

\section{No-Go: Retaining an agent's control}

\noindent

\noindent A process theoretic perspective highlights how de-coupling classical agents from the physical theory is a somewhat odd thing to do: if everything else is modelled explicitly, there is no real reason to believe that agents should not. In other approaches this may appear less odd, but it is only because agents held an external position in the first place.
Thus said, de-coupling agents from the theory does not solve our problems: as it turns out, the agent would still
be able to superluminally signal. Indeed, all they have to do is apply the time-reverse of a non-demolition Bell measurement to deterministically realize the projector on the Bell state (where we have exploited the Choi isomorphism to represent all Bell-basis projectors in terms of the Bell state):
\[
\input{measdiag4.tikz}\quad=\quad\input{measdiag7.tikz}
\]
This in turn result in the possibility of performing deterministic quantum teleportation \cite{Kindergarten}:
\[
\input{bell-teleportation.tikz}
\]
Because the agent is now in possession of states other than noise, deterministic quantum teleportation immediately leads to a violation of no-signalling. For example, they could use the time-reverse of the demolition Pauli Z measurement to selectively prepare computation basis states, which they would then deterministically teleport:
\[
\input{comp-basis-prep.tikz}
\]
We conclude that retaining an agent's full classical control 
does nothing to help restore time-symmetry: the original causal theory is non-signalling, while the retro-causal theory is signalling,  so the symmetry is 
broken.

\section{Discussion and conclusion}
%

\noindent The physical interpretation of our main result depends on a choice of ontological stance on states and processes. In modern physics, the standard assumption is that the state of a physical system describes its actual ``state of affairs'', evolved and studied by solving appropriate sets of equations: for example, the ontic state of a physical system in classical Hamiltonian mechanics is fully described by a considering a point in its phase space, and can be evolved in time by solving Hamilton's Equations. In this sense, a retro-causal process theory is truly eternal noise, and hence utterly trivial.

One could  
consider the alternative, long-standing viewpoint that it is processes, rather than states, that fully characterize the ``state of affairs'' of physical systems: it is \emph{becoming} that matters, rather than \emph{being} \cite{Whitehead,verelst1999early}.  In this case, the retro-causal theory is no longer trivial, as it ends up having as many processes as the causal theory originally had.
It would therefore be tempting to conclude that a causal process theory and its time-reverse have the same ``informational content'', with the only distinction that inputs and outputs have exchanged roles \cite{d'ariano2014determinism}.
However, 
these two process theories are not at all equivalent: 
the processes that can be applied to a particular system will be different for the  two theories, and in particular the states never coincide!


The results we present here
seem to be in conflict with recent time-symmetric approaches to quantum theory,
such as those by \cite{aharonov2010time,oreshkov2015operational,oreshkov2016operational}. The resolution  of this apparent conflict lies in the observation that those time-symmetric theories are not terminal in their most general formulation: they are not \textit{a priori} compatible with relativity, at least not in the compositional, process-theoretic sense laid out by \cite{IFF}. Instead, asymmetric boundary conditions must be imposed to ensure compatibility with causal structure (e.g.~to forbid superluminal signalling). It was shown in \cite{oreshkov2015operational} that suitable boundary conditions and dynamics can be chosen so that process terminality is recovered: in that case, the results of the time-symmetric framework coincide with those presented in this work. To us, however, this entire setup seems unnecessarily complicated: after all, relativistic principles do not make any obvious statements about the boundary conditions of the universe. In contrast, the terminality condition that we impose is directly related to the causal structure of relativity, and as such we deem it to be a far more natural constraint.

\section {Acknowledgements}

The result of this paper emerged in a discussion following a debate led by Ognyan Oreshkov on time-symmetry in physics. Ognyan also pointed out some shortcomings in the first publicized version of this paper with regard to the presentation of his results in \cite{oreshkov2015operational}. We also thank Robert Raussendorf and Rui Soares Barbosa for discussions which greatly improved the presentation of our results. This project/publication  was made possible through the support of a grant  from the John Templeton Foundation. The opinions expressed in this publication are those of the author(s) and do not necessarily reflect the views of the John Templeton Foundation. JHS is funded by Perimeter Institute for Theoretical Physics. Research at Perimeter Institute is supported by the Government of Canada through the Department of Innovation, Science and Economic Development Canada and by the Province of Ontario through the Ministry of Research, Innovation and Science.

\bibliographystyle{apsrev}
\bibliography{main}

\end{document}

%% file: preamble.tex
\usepackage{rotating}
\usepackage{floatpag}
\usepackage{xifthen}
\usepackage[normalem]{ulem}

\usepackage{bm}

\newcommand{\warningsign}{\raisebox{1pt}{\fontencoding{U}\fontfamily{futs}\selectfont\char 66\relax}\xspace}

\ifthenelse{\isundefined{\showoptional}}{
  \newcommand{\showoptional}{1}
}{}

\ifthenelse{\isundefined{\ismain}}{
  \newcommand{\ismain}{0}
}{}

\usepackage{hyperref}
\usepackage{amsmath,amsthm,amssymb}
\usepackage{wasysym}

\usepackage{xspace,enumerate,color,epsfig}
\usepackage{graphicx}
\graphicspath{{.}{./figures/}}

\usepackage{tikzfig}
\usepackage{stmaryrd}
\usepackage{docmute}
\usepackage{keycommand}

\input{defs.tex}
\input{tikzstyles.tex}

\let\olddagger\dagger
\renewcommand{\dagger}{\ensuremath{\olddagger}\xspace}

\usepackage{makeidx}
\makeindex

\theoremstyle{definition}
\newtheorem{theorem}{Theorem}
\newtheorem*{theorem*}{Theorem}

\newtheorem{definition}[theorem]{Definition}

\newtheorem{example*}[theorem]{Example*}
\newtheorem{examples*}[theorem]{Examples*}

\newtheorem{remark*}[theorem]{Remark*}

\newtheorem{exer*}[theorem]{Exercise*}

\newkeycommand{\moral}[width=11cm][1]{\begin{center}
\fbox{\ \parbox{\commandkey{width}}{\centering #1\vphantom{Xy}}\ }
\end{center}}

\newkeycommand{\morallong}[width=11cm][1]{\par\medskip\noindent
\centerline{\fbox{\ \parbox{\commandkey{width}}{\centering #1\vphantom{Xy}}\ }}
\par\medskip\noindent}

\usepackage{color}
\def\bR{\begin{color}{red}}
\def\bB{\begin{color}{blue}}
\def\bM{\begin{color}{magenta}}
\def\bC{\begin{color}{cyan}}
\def\bW{\begin{color}{white}}
\def\bBl{\begin{color}{black}}
\def\bG{\begin{color}{green}}
\def\bY{\begin{color}{yellow}}
\def\jR{\begin{color}{magenta}}
\def\jB{\begin{color}{cyan}}
\def\e{\end{color}\xspace}
\newcommand{\bit}{\begin{itemize}}
\newcommand{\eit}{\end{itemize}\par\noindent}
\newcommand{\ben}{\begin{enumerate}}
\newcommand{\een}{\end{enumerate}\par\noindent}
\newcommand{\beq}{\begin{equation}}
\newcommand{\eeq}{\end{equation}\par\noindent}
\newcommand{\beqa}{\begin{eqnarray*}}
\newcommand{\eeqa}{\end{eqnarray*}\par\noindent}
\newcommand{\beqn}{\begin{eqnarray}}
\newcommand{\eeqn}{\end{eqnarray}\par\noindent}

\newcommand{\timerev}[1]{\reflectbox{\rotatebox[origin=c]{180}{$#1$}}}

%% file: defs.tex
\newcommand{\past}{\textbf{past}}

\newcommand{\inp}[1]{#1^{i}}
\newcommand{\outp}[1]{#1^{o}}

%% file: tikzstyles.tex
\tikzstyle{env}=[copoint,regular polygon rotate=0,minimum width=0.2cm, fill=black]

\tikzstyle{probs}=[shape=semicircle,fill=white,draw=black,shape border rotate=180,minimum width=1.2cm]

\tikzstyle{every picture}=[baseline=-0.25em,scale=0.5]
\tikzstyle{dotpic}=[] 
\tikzstyle{diredges}=[every to/.style={diredge}]
\tikzstyle{math matrix}=[matrix of math nodes,left delimiter=(,right delimiter=),inner sep=2pt,column sep=1em,row sep=0.5em,nodes={inner sep=0pt},text height=1.5ex, text depth=0.25ex]

\tikzstyle{inline text}=[text height=1.5ex, text depth=0.25ex,yshift=0.5mm]
\tikzstyle{label}=[font=\footnotesize,text height=1.5ex, text depth=0.25ex,yshift=0.5mm]
\tikzstyle{left label}=[label,anchor=east,xshift=1.5mm]
\tikzstyle{right label}=[label,anchor=west,xshift=-1.5mm]

\tikzstyle{braceedge}=[decorate,decoration={brace,amplitude=2mm,raise=-1mm}]
\tikzstyle{small braceedge}=[decorate,decoration={brace,amplitude=1mm,raise=-1mm}]

\tikzstyle{doubled}=[line width=1.6pt] 
\tikzstyle{boldedge}=[doubled,shorten <=-0.17mm,shorten >=-0.17mm]
\tikzstyle{boldedgegray}=[doubled,gray,shorten <=-0.17mm,shorten >=-0.17mm]
\tikzstyle{singleedgegray}=[gray]

\tikzstyle{semidoubled}=[line width=1.4pt] 
\tikzstyle{semiboldedgegray}=[semidoubled,gray,shorten <=-0.17mm,shorten >=-0.17mm]

\tikzstyle{boxedge}=[semiboldedgegray]

\tikzstyle{boldedgedashed}=[very thick,dashed,shorten <=-0.17mm,shorten >=-0.17mm]
\tikzstyle{vboldedgedashed}=[doubled,dashed,shorten <=-0.17mm,shorten >=-0.17mm]
\tikzstyle{left hook arrow}=[left hook-latex]
\tikzstyle{right hook arrow}=[right hook-latex]
\tikzstyle{sembracket}=[line width=0.5pt,shorten <=-0.07mm,shorten >=-0.07mm]

\tikzstyle{causal edge}=[->,thick,gray]
\tikzstyle{causal nondir}=[thick,gray]
\tikzstyle{timeline}=[thick,gray, dashed]

\tikzstyle{cedge}=[<->,thick,gray!70!white]

\tikzstyle{empty diagram}=[draw=gray!40!white,dashed,shape=rectangle,minimum width=1cm,minimum height=1cm]
\tikzstyle{empty diagram small}=[draw=gray!50!white,dashed,shape=rectangle,minimum width=0.6cm,minimum height=0.5cm]

\tikzstyle{dot}=[inner sep=0mm,minimum width=2mm,minimum height=2mm,draw,shape=circle]  
\tikzstyle{Wsquare}=[white dot, shape=regular polygon, rounded corners=0.8 mm, minimum size=3.3 mm, regular polygon sides=3, outer sep=-0.2mm]
\tikzstyle{Wsquareadj}=[white dot, shape=regular polygon, rounded corners=0.8 mm, minimum size=3.3 mm, regular polygon sides=3, outer sep=-0.2mm, regular polygon rotate=180]
\tikzstyle{ddot}=[inner sep=0mm, doubled, minimum width=2.5mm,minimum height=2.5mm,draw,shape=circle]

\tikzstyle{black dot}=[dot,fill=black]
\tikzstyle{white dot}=[dot,fill=white,,text depth=-0.2mm]
\tikzstyle{white Wsquare}=[Wsquare,fill=gray,,text depth=-0.2mm]
\tikzstyle{white Wsquareadj}=[Wsquareadj,fill=white,,text depth=-0.2mm]
\tikzstyle{green dot}=[white dot] 
\tikzstyle{gray dot}=[dot,fill=gray!40!white,,text depth=-0.2mm]
\tikzstyle{red dot}=[gray dot] 

\tikzstyle{black ddot}=[ddot,fill=black]
\tikzstyle{white ddot}=[ddot,fill=white]
\tikzstyle{gray ddot}=[ddot,fill=gray!40!white]

\tikzstyle{gray edge}=[gray!60!white]

\tikzstyle{small dot}=[inner sep=0.5mm,minimum width=0pt,minimum height=0pt,draw,shape=circle]

\tikzstyle{small black dot}=[small dot,fill=black]
\tikzstyle{small white dot}=[small dot,fill=white]
\tikzstyle{small gray dot}=[small dot,fill=gray!40!white]

\tikzstyle{causal dot}=[inner sep=0.4mm,minimum width=0pt,minimum height=0pt,draw=white,shape=circle,fill=gray!40!white]

\tikzstyle{phase dimensions}=[minimum size=5mm,font=\footnotesize,rectangle,rounded corners=2.5mm,inner sep=0.2mm,outer sep=-2mm]
\tikzstyle{dphase dimensions}=[minimum size=5mm,font=\footnotesize,rectangle,rounded corners=2.5mm,inner sep=0.2mm,outer sep=-2mm]

\tikzstyle{white phase dot}=[dot,fill=white,phase dimensions]
\tikzstyle{white phase ddot}=[ddot,fill=white,dphase dimensions]

\tikzstyle{white rect ddot}=[draw=black,fill=white,doubled,minimum size=5mm,font=\footnotesize,rectangle,rounded corners=2.5mm,inner sep=0.2mm]
\tikzstyle{gray rect ddot}=[draw=black,fill=gray!40!white,doubled,minimum size=6mm,font=\footnotesize,rectangle,rounded corners=3mm]

\tikzstyle{gray phase dot}=[dot,fill=gray!40!white,phase dimensions]
\tikzstyle{gray phase ddot}=[ddot,fill=gray!40!white,dphase dimensions]
\tikzstyle{grey phase dot}=[gray phase dot]
\tikzstyle{grey phase ddot}=[gray phase ddot]

\tikzstyle{small phase dimensions}=[minimum size=4mm,font=\tiny,rectangle,rounded corners=2mm,inner sep=0.2mm,outer sep=-2mm]
\tikzstyle{small dphase dimensions}=[minimum size=4mm,font=\tiny,rectangle,rounded corners=2mm,inner sep=0.2mm,outer sep=-2mm]

\tikzstyle{small gray phase dot}=[dot,fill=gray!40!white,small phase dimensions]
\tikzstyle{small gray phase ddot}=[ddot,fill=gray!40!white,small dphase dimensions]

\tikzstyle{small map}=[draw,shape=rectangle,minimum height=4mm,minimum width=4mm,fill=white]

\tikzstyle{cnot}=[fill=white,shape=circle,inner sep=-1.4pt]

\tikzstyle{asym hadamard}=[fill=white,draw,shape=NEbox,inner sep=0.6mm,font=\footnotesize,minimum height=4mm]
\tikzstyle{asym hadamard conj}=[fill=white,draw,shape=NWbox,inner sep=0.6mm,font=\footnotesize,minimum height=4mm]
\tikzstyle{asym hadamard dag}=[fill=white,draw,shape=SEbox,inner sep=0.6mm,font=\footnotesize,minimum height=4mm]

\tikzstyle{hadamard}=[fill=white,draw,inner sep=0.6mm,font=\footnotesize,minimum height=4mm,minimum width=4mm]
\tikzstyle{small hadamard}=[fill=white,draw,inner sep=0.6mm,minimum height=1.5mm,minimum width=1.5mm]
\tikzstyle{small hadamard rotate}=[small hadamard,rotate=45]
\tikzstyle{dhadamard}=[hadamard,doubled]
\tikzstyle{small dhadamard}=[small hadamard,doubled]
\tikzstyle{small dhadamard rotate}=[small hadamard rotate,doubled]
\tikzstyle{antipode}=[white dot,inner sep=0.3mm,font=\footnotesize]

\tikzstyle{scalar}=[diamond,draw,inner sep=0.5pt,font=\small]
\tikzstyle{dscalar}=[diamond,doubled, draw,inner sep=0.5pt,font=\small]

\tikzstyle{small box}=[rectangle,inline text,fill=white,draw,minimum height=5mm,yshift=-0.5mm,minimum width=5mm,font=\small]
\tikzstyle{small gray box}=[small box,fill=gray!30]
\tikzstyle{medium box}=[rectangle,inline text,fill=white,draw,minimum height=5mm,yshift=-0.5mm,minimum width=10mm,font=\small]
\tikzstyle{square box}=[small box] 
\tikzstyle{medium gray box}=[small box,fill=gray!30]
\tikzstyle{semilarge box}=[rectangle,inline text,fill=white,draw,minimum height=5mm,yshift=-0.5mm,minimum width=12.5mm,font=\small]
\tikzstyle{large box}=[rectangle,inline text,fill=white,draw,minimum height=5mm,yshift=-0.5mm,minimum width=15mm,font=\small]
\tikzstyle{large gray box}=[small box,fill=gray!30]

\tikzstyle{Bayes box}=[rectangle,fill=black,draw, minimum height=3mm, minimum width=3mm]

\tikzstyle{gray square point}=[small box,fill=gray!50]

\tikzstyle{dphase box white}=[dhadamard]
\tikzstyle{dphase box gray}=[dhadamard,fill=gray!50!white]
\tikzstyle{phase box white}=[hadamard]
\tikzstyle{phase box gray}=[hadamard,fill=gray!50!white]

\tikzstyle{point}=[regular polygon,regular polygon sides=3,draw,scale=0.75,inner sep=-0.5pt,minimum width=9mm,fill=white,regular polygon rotate=180]
\tikzstyle{point nosep}=[regular polygon,regular polygon sides=3,draw,scale=0.75,inner sep=-2pt,minimum width=9mm,fill=white,regular polygon rotate=180]
\tikzstyle{copoint}=[regular polygon,regular polygon sides=3,draw,scale=0.75,inner sep=-0.5pt,minimum width=9mm,fill=white]
\tikzstyle{dpoint}=[point,doubled]
\tikzstyle{dcopoint}=[copoint,doubled]

\tikzstyle{pointgrow}=[shape=cornerpoint,kpoint common,scale=0.75,inner sep=3pt]
\tikzstyle{pointgrow dag}=[shape=cornercopoint,kpoint common,scale=0.75,inner sep=3pt]

\tikzstyle{wide copoint}=[fill=white,draw,shape=isosceles triangle,shape border rotate=90,isosceles triangle stretches=true,inner sep=0pt,minimum width=1.5cm,minimum height=6.12mm]
\tikzstyle{wide point}=[fill=white,draw,shape=isosceles triangle,shape border rotate=-90,isosceles triangle stretches=true,inner sep=0pt,minimum width=1.5cm,minimum height=6.12mm,yshift=-0.0mm]
\tikzstyle{wide point plus}=[fill=white,draw,shape=isosceles triangle,shape border rotate=-90,isosceles triangle stretches=true,inner sep=0pt,minimum width=1.74cm,minimum height=7mm,yshift=-0.0mm]

\tikzstyle{wide dpoint}=[fill=white,doubled,draw,shape=isosceles triangle,shape border rotate=-90,isosceles triangle stretches=true,inner sep=0pt,minimum width=1.5cm,minimum height=6.12mm,yshift=-0.0mm]

\tikzstyle{tinypoint}=[regular polygon,regular polygon sides=3,draw,scale=0.55,inner sep=-0.15pt,minimum width=6mm,fill=white,regular polygon rotate=180] 

\tikzstyle{white point}=[point]
\tikzstyle{white dpoint}=[dpoint]
\tikzstyle{green point}=[white point] 
\tikzstyle{white copoint}=[copoint]
\tikzstyle{gray point}=[point,fill=gray!40!white]
\tikzstyle{gray dpoint}=[gray point,doubled]
\tikzstyle{red point}=[gray point] 
\tikzstyle{gray copoint}=[copoint,fill=gray!40!white]
\tikzstyle{gray dcopoint}=[gray copoint,doubled]

\tikzstyle{white point guide}=[regular polygon,regular polygon sides=3,font=\scriptsize,draw,scale=0.65,inner sep=-0.5pt,minimum width=9mm,fill=white,regular polygon rotate=180]

\tikzstyle{black point}=[point,fill=black,font=\color{white}]
\tikzstyle{black copoint}=[copoint,fill=black,font=\color{white}]

\tikzstyle{tiny gray point}=[tinypoint,fill=gray!40!white]

\tikzstyle{diredge}=[->]
\tikzstyle{ddiredge}=[<->]
\tikzstyle{rdiredge}=[<-]
\tikzstyle{thickdiredge}=[->, very thick]
\tikzstyle{pointer edge}=[->,very thick,gray]
\tikzstyle{pointer edge part}=[very thick,gray]
\tikzstyle{dashed edge}=[dashed]
\tikzstyle{thick dashed edge}=[very thick,dashed]
\tikzstyle{thick gray dashed edge}=[thick dashed edge,gray!40]
\tikzstyle{thick map edge}=[very thick,|->]

\makeatletter
\newcommand{\boxshape}[3]{%
\pgfdeclareshape{#1}{
\inheritsavedanchors[from=rectangle] 
\inheritanchorborder[from=rectangle]
\inheritanchor[from=rectangle]{center}
\inheritanchor[from=rectangle]{north}
\inheritanchor[from=rectangle]{south}
\inheritanchor[from=rectangle]{west}
\inheritanchor[from=rectangle]{east}
\backgroundpath{
\southwest \pgf@xa=\pgf@x \pgf@ya=\pgf@y
\northeast \pgf@xb=\pgf@x \pgf@yb=\pgf@y

\@tempdima=#2
\@tempdimb=#3

\pgfpathmoveto{\pgfpoint{\pgf@xa - 5pt + \@tempdima}{\pgf@ya}}
\pgfpathlineto{\pgfpoint{\pgf@xa - 5pt - \@tempdima}{\pgf@yb}}
\pgfpathlineto{\pgfpoint{\pgf@xb + 5pt + \@tempdimb}{\pgf@yb}}
\pgfpathlineto{\pgfpoint{\pgf@xb + 5pt - \@tempdimb}{\pgf@ya}}
\pgfpathlineto{\pgfpoint{\pgf@xa - 5pt + \@tempdima}{\pgf@ya}}
\pgfpathclose
}
}}

\boxshape{NEbox}{0pt}{5pt}
\boxshape{SEbox}{0pt}{-5pt}
\boxshape{NWbox}{5pt}{0pt}
\boxshape{SWbox}{-5pt}{0pt}
\boxshape{EBox}{-3pt}{3pt}
\boxshape{WBox}{3pt}{-3pt}
\makeatother

\tikzstyle{cloud}=[shape=cloud,draw,minimum width=1.5cm,minimum height=1.5cm]

\tikzstyle{map}=[draw,shape=NEbox,inner sep=2pt,minimum height=6mm,fill=white]
\tikzstyle{dashedmap}=[draw,dashed,shape=NEbox,inner sep=2pt,minimum height=6mm,fill=white]
\tikzstyle{mapdag}=[draw,shape=SEbox,inner sep=2pt,minimum height=6mm,fill=white]
\tikzstyle{mapadj}=[draw,shape=SEbox,inner sep=2pt,minimum height=6mm,fill=white]
\tikzstyle{maptrans}=[draw,shape=SWbox,inner sep=2pt,minimum height=6mm,fill=white]
\tikzstyle{mapconj}=[draw,shape=NWbox,inner sep=2pt,minimum height=6mm,fill=white]

\tikzstyle{medium map}=[draw,shape=NEbox,inner sep=2pt,minimum height=6mm,fill=white,minimum width=7mm]
\tikzstyle{medium map dag}=[draw,shape=SEbox,inner sep=2pt,minimum height=6mm,fill=white,minimum width=7mm]
\tikzstyle{medium map adj}=[draw,shape=SEbox,inner sep=2pt,minimum height=6mm,fill=white,minimum width=7mm]
\tikzstyle{medium map trans}=[draw,shape=SWbox,inner sep=2pt,minimum height=6mm,fill=white,minimum width=7mm]
\tikzstyle{medium map conj}=[draw,shape=NWbox,inner sep=2pt,minimum height=6mm,fill=white,minimum width=7mm]
\tikzstyle{semilarge map}=[draw,shape=NEbox,inner sep=2pt,minimum height=6mm,fill=white,minimum width=9.5mm]
\tikzstyle{semilarge map trans}=[draw,shape=SWbox,inner sep=2pt,minimum height=6mm,fill=white,minimum width=9.5mm]
\tikzstyle{semilarge map adj}=[draw,shape=SEbox,inner sep=2pt,minimum height=6mm,fill=white,minimum width=9.5mm]
\tikzstyle{semilarge map dag}=[draw,shape=SEbox,inner sep=2pt,minimum height=6mm,fill=white,minimum width=9.5mm]
\tikzstyle{semilarge map conj}=[draw,shape=NWbox,inner sep=2pt,minimum height=6mm,fill=white,minimum width=9.5mm]
\tikzstyle{large map}=[draw,shape=NEbox,inner sep=2pt,minimum height=6mm,fill=white,minimum width=12mm]
\tikzstyle{large map conj}=[draw,shape=NWbox,inner sep=2pt,minimum height=6mm,fill=white,minimum width=12mm]
\tikzstyle{very large map}=[draw,shape=NEbox,inner sep=2pt,minimum height=6mm,fill=white,minimum width=17mm]

\tikzstyle{medium dmap}=[draw,doubled,shape=NEbox,inner sep=2pt,minimum height=6mm,fill=white,minimum width=7mm]
\tikzstyle{medium dmap dag}=[draw,doubled,shape=SEbox,inner sep=2pt,minimum height=6mm,fill=white,minimum width=7mm]
\tikzstyle{medium dmap adj}=[draw,doubled,shape=SEbox,inner sep=2pt,minimum height=6mm,fill=white,minimum width=7mm]
\tikzstyle{medium dmap trans}=[draw,doubled,shape=SWbox,inner sep=2pt,minimum height=6mm,fill=white,minimum width=7mm]
\tikzstyle{medium dmap conj}=[draw,doubled,shape=NWbox,inner sep=2pt,minimum height=6mm,fill=white,minimum width=7mm]
\tikzstyle{semilarge dmap}=[draw,doubled,shape=NEbox,inner sep=2pt,minimum height=6mm,fill=white,minimum width=9.5mm]
\tikzstyle{semilarge dmap trans}=[draw,doubled,shape=SWbox,inner sep=2pt,minimum height=6mm,fill=white,minimum width=9.5mm]
\tikzstyle{semilarge dmap adj}=[draw,doubled,shape=SEbox,inner sep=2pt,minimum height=6mm,fill=white,minimum width=9.5mm]
\tikzstyle{semilarge dmap dag}=[draw,doubled,shape=SEbox,inner sep=2pt,minimum height=6mm,fill=white,minimum width=9.5mm]
\tikzstyle{semilarge dmap conj}=[draw,doubled,shape=NWbox,inner sep=2pt,minimum height=6mm,fill=white,minimum width=9.5mm]
\tikzstyle{large dmap}=[draw,doubled,shape=NEbox,inner sep=2pt,minimum height=6mm,fill=white,minimum width=12mm]
\tikzstyle{large dmap conj}=[draw,doubled,shape=NWbox,inner sep=2pt,minimum height=6mm,fill=white,minimum width=12mm]
\tikzstyle{large dmap trans}=[draw,doubled,shape=SWbox,inner sep=2pt,minimum height=6mm,fill=white,minimum width=12mm]
\tikzstyle{large dmap adj}=[draw,doubled,shape=SEbox,inner sep=2pt,minimum height=6mm,fill=white,minimum width=12mm]
\tikzstyle{large dmap dag}=[draw,doubled,shape=SEbox,inner sep=2pt,minimum height=6mm,fill=white,minimum width=12mm]
\tikzstyle{very large dmap}=[draw,doubled,shape=NEbox,inner sep=2pt,minimum height=6mm,fill=white,minimum width=19.5mm]

\tikzstyle{muxbox}=[draw,shape=rectangle,minimum height=3mm,minimum width=3mm,fill=white]
\tikzstyle{dmuxbox}=[muxbox,doubled]

\tikzstyle{box}=[draw,shape=rectangle,inner sep=2pt,minimum height=6mm,minimum width=6mm,fill=white]
\tikzstyle{dbox}=[draw,doubled,shape=rectangle,inner sep=2pt,minimum height=6mm,minimum width=6mm,fill=white]
\tikzstyle{dmap}=[draw,doubled,shape=NEbox,inner sep=2pt,minimum height=6mm,fill=white]
\tikzstyle{dmapdag}=[draw,doubled,shape=SEbox,inner sep=2pt,minimum height=6mm,fill=white]
\tikzstyle{dmapadj}=[draw,doubled,shape=SEbox,inner sep=2pt,minimum height=6mm,fill=white]
\tikzstyle{dmaptrans}=[draw,doubled,shape=SWbox,inner sep=2pt,minimum height=6mm,fill=white]
\tikzstyle{dmapconj}=[draw,doubled,shape=NWbox,inner sep=2pt,minimum height=6mm,fill=white]

\tikzstyle{ddmap}=[draw,doubled,dashed,shape=NEbox,inner sep=2pt,minimum height=6mm,fill=white]
\tikzstyle{ddmapdag}=[draw,doubled,dashed,shape=SEbox,inner sep=2pt,minimum height=6mm,fill=white]
\tikzstyle{ddmapadj}=[draw,doubled,dashed,shape=SEbox,inner sep=2pt,minimum height=6mm,fill=white]
\tikzstyle{ddmaptrans}=[draw,doubled,dashed,shape=SWbox,inner sep=2pt,minimum height=6mm,fill=white]
\tikzstyle{ddmapconj}=[draw,doubled,dashed,shape=NWbox,inner sep=2pt,minimum height=6mm,fill=white]

\boxshape{sNEbox}{0pt}{3pt}
\boxshape{sSEbox}{0pt}{-3pt}
\boxshape{sNWbox}{3pt}{0pt}
\boxshape{sSWbox}{-3pt}{0pt}
\tikzstyle{smap}=[draw,shape=sNEbox,fill=white]
\tikzstyle{smapdag}=[draw,shape=sSEbox,fill=white]
\tikzstyle{smapadj}=[draw,shape=sSEbox,fill=white]
\tikzstyle{smaptrans}=[draw,shape=sSWbox,fill=white]
\tikzstyle{smapconj}=[draw,shape=sNWbox,fill=white]

\tikzstyle{dsmap}=[draw,dashed,shape=sNEbox,fill=white]
\tikzstyle{dsmapdag}=[draw,dashed,shape=sSEbox,fill=white]
\tikzstyle{dsmaptrans}=[draw,dashed,shape=sSWbox,fill=white]
\tikzstyle{dsmapconj}=[draw,dashed,shape=sNWbox,fill=white]

\boxshape{mNEbox}{0pt}{10pt}
\boxshape{mSEbox}{0pt}{-10pt}
\boxshape{mNWbox}{10pt}{0pt}
\boxshape{mSWbox}{-10pt}{0pt}
\tikzstyle{mmap}=[draw,shape=mNEbox]
\tikzstyle{mmapdag}=[draw,shape=mSEbox]
\tikzstyle{mmaptrans}=[draw,shape=mSWbox]
\tikzstyle{mmapconj}=[draw,shape=mNWbox]

\tikzstyle{mmapgray}=[draw,fill=gray!40!white,shape=mNEbox]
\tikzstyle{smapgray}=[draw,fill=gray!40!white,shape=sNEbox]

\makeatletter

\pgfdeclareshape{cornerpoint}{
\inheritsavedanchors[from=rectangle] 
\inheritanchorborder[from=rectangle]
\inheritanchor[from=rectangle]{center}
\inheritanchor[from=rectangle]{north}
\inheritanchor[from=rectangle]{south}
\inheritanchor[from=rectangle]{west}
\inheritanchor[from=rectangle]{east}
\backgroundpath{
\southwest \pgf@xa=\pgf@x \pgf@ya=\pgf@y
\northeast \pgf@xb=\pgf@x \pgf@yb=\pgf@y

\pgfmathsetmacro{\pgf@shorten@left}{\pgfkeysvalueof{/tikz/shorten left}}
\pgfmathsetmacro{\pgf@shorten@right}{\pgfkeysvalueof{/tikz/shorten right}}

\pgfpathmoveto{\pgfpoint{0.5 * (\pgf@xa + \pgf@xb)}{\pgf@ya - 5pt}}
\pgfpathlineto{\pgfpoint{\pgf@xa - 8pt + \pgf@shorten@left}{\pgf@yb - 1.5 * \pgf@shorten@left}}
\pgfpathlineto{\pgfpoint{\pgf@xa - 8pt + \pgf@shorten@left}{\pgf@yb}}
\pgfpathlineto{\pgfpoint{\pgf@xb + 8pt - \pgf@shorten@right}{\pgf@yb}}
\pgfpathlineto{\pgfpoint{\pgf@xb + 8pt - \pgf@shorten@right}{\pgf@yb - 1.5 * \pgf@shorten@right}}
\pgfpathclose
}
}

\pgfdeclareshape{cornercopoint}{
\inheritsavedanchors[from=rectangle] 
\inheritanchorborder[from=rectangle]
\inheritanchor[from=rectangle]{center}
\inheritanchor[from=rectangle]{north}
\inheritanchor[from=rectangle]{south}
\inheritanchor[from=rectangle]{west}
\inheritanchor[from=rectangle]{east}
\backgroundpath{
\southwest \pgf@xa=\pgf@x \pgf@ya=\pgf@y
\northeast \pgf@xb=\pgf@x \pgf@yb=\pgf@y

\pgfmathsetmacro{\pgf@shorten@left}{\pgfkeysvalueof{/tikz/shorten left}}
\pgfmathsetmacro{\pgf@shorten@right}{\pgfkeysvalueof{/tikz/shorten right}}

\pgfpathmoveto{\pgfpoint{0.5 * (\pgf@xa + \pgf@xb)}{\pgf@yb + 5pt}}
\pgfpathlineto{\pgfpoint{\pgf@xa - 8pt + \pgf@shorten@left}{\pgf@ya + 1.5 * \pgf@shorten@left}}
\pgfpathlineto{\pgfpoint{\pgf@xa - 8pt + \pgf@shorten@left}{\pgf@ya}}
\pgfpathlineto{\pgfpoint{\pgf@xb + 8pt - \pgf@shorten@right}{\pgf@ya}}
\pgfpathlineto{\pgfpoint{\pgf@xb + 8pt - \pgf@shorten@right}{\pgf@ya + 1.5 * \pgf@shorten@right}}
\pgfpathclose
}
}

\makeatother

\pgfkeyssetvalue{/tikz/shorten left}{0pt}
\pgfkeyssetvalue{/tikz/shorten right}{0pt}

\tikzstyle{kpoint common}=[draw,fill=white,inner sep=1pt,minimum height=4mm]
\tikzstyle{kpoint sc}=[shape=cornerpoint,kpoint common]
\tikzstyle{kpoint adjoint sc}=[shape=cornercopoint,kpoint common]
\tikzstyle{kpoint}=[shape=cornerpoint,shorten left=5pt,kpoint common]
\tikzstyle{kpoint adjoint}=[shape=cornercopoint,shorten left=5pt,kpoint common]
\tikzstyle{kpoint conjugate}=[shape=cornerpoint,shorten right=5pt,kpoint common]
\tikzstyle{kpoint transpose}=[shape=cornercopoint,shorten right=5pt,kpoint common]
\tikzstyle{kpoint symm}=[shape=cornerpoint,shorten left=5pt,shorten right=5pt,kpoint common]

\tikzstyle{wide kpoint sc}=[shape=cornerpoint,kpoint common, minimum width=1 cm]
\tikzstyle{wide kpointdag sc}=[shape=cornercopoint,kpoint common, minimum width=1 cm]

\tikzstyle{black kpoint}=[shape=cornerpoint,shorten left=5pt,kpoint common,fill=black,font=\color{white}]

\tikzstyle{black kpoint sm}=[shape=cornerpoint,shorten left=5pt,kpoint common,fill=black,font=\color{white},scale=0.75]

\tikzstyle{black kpoint adjoint}=[shape=cornercopoint,shorten left=5pt,kpoint common,fill=black,font=\color{white}]
\tikzstyle{black kpointadj}=[shape=cornercopoint,shorten left=5pt,kpoint common,fill=black,font=\color{white}]

\tikzstyle{black kpointadj sm}=[shape=cornercopoint,shorten left=5pt,kpoint common,fill=black,font=\color{white},scale=0.75]

\tikzstyle{black dkpoint}=[shape=cornerpoint,shorten left=5pt,kpoint common,fill=black, doubled,font=\color{white}]
\tikzstyle{black dkpoint adjoint}=[shape=cornercopoint,shorten left=5pt,kpoint common,fill=black, doubled,font=\color{white}]
\tikzstyle{black dkpointadj}=[shape=cornercopoint,shorten left=5pt,kpoint common,fill=black, doubled,font=\color{white}]

\tikzstyle{black dkpoint sm}=[shape=cornerpoint,shorten left=5pt,kpoint common,fill=black, doubled,font=\color{white},scale=0.75]
\tikzstyle{black dkpointadj sm}=[shape=cornercopoint,shorten left=5pt,kpoint common,fill=black, doubled,font=\color{white},scale=0.75] 

\tikzstyle{kpointdag}=[kpoint adjoint]
\tikzstyle{kpointadj}=[kpoint adjoint]
\tikzstyle{kpointconj}=[kpoint conjugate]
\tikzstyle{kpointtrans}=[kpoint transpose]

\tikzstyle{big kpoint}=[kpoint, minimum width=1.2 cm, minimum height=8mm, inner sep=4pt, text depth=3mm]

\tikzstyle{wide kpoint}=[kpoint, minimum width=1 cm, inner sep=2pt]
\tikzstyle{wide kpointdag}=[kpointdag, minimum width=1 cm, inner sep=2pt]
\tikzstyle{wide kpointconj}=[kpointconj, minimum width=1 cm, inner sep=2pt]
\tikzstyle{wide kpointtrans}=[kpointtrans, minimum width=1 cm, inner sep=2pt]

\tikzstyle{wider kpoint}=[kpoint, minimum width=1.25 cm, inner sep=2pt]
\tikzstyle{wider kpointdag}=[kpointdag, minimum width=1.25 cm, inner sep=2pt]
\tikzstyle{wider kpointconj}=[kpointconj, minimum width=1.25 cm, inner sep=2pt]
\tikzstyle{wider kpointtrans}=[kpointtrans, minimum width=1.25 cm, inner sep=2pt]

\tikzstyle{gray kpoint}=[kpoint,fill=gray!50!white]
\tikzstyle{gray kpointdag}=[kpointdag,fill=gray!50!white]
\tikzstyle{gray kpointadj}=[kpointadj,fill=gray!50!white]
\tikzstyle{gray kpointconj}=[kpointconj,fill=gray!50!white]
\tikzstyle{gray kpointtrans}=[kpointtrans,fill=gray!50!white]

\tikzstyle{gray dkpoint}=[kpoint,fill=gray!50!white,doubled]
\tikzstyle{gray dkpointdag}=[kpointdag,fill=gray!50!white,doubled]
\tikzstyle{gray dkpointadj}=[kpointadj,fill=gray!50!white,doubled]
\tikzstyle{gray dkpointconj}=[kpointconj,fill=gray!50!white,doubled]
\tikzstyle{gray dkpointtrans}=[kpointtrans,fill=gray!50!white,doubled]

\tikzstyle{white label}=[draw,fill=white,rectangle,inner sep=0.7 mm]
\tikzstyle{gray label}=[draw,fill=gray!50!white,rectangle,inner sep=0.7 mm]
\tikzstyle{black label}=[draw,fill=black,rectangle,inner sep=0.7 mm]

\tikzstyle{dkpoint}=[kpoint,doubled]
\tikzstyle{wide dkpoint}=[wide kpoint,doubled]
\tikzstyle{dkpointdag}=[kpoint adjoint,doubled]
\tikzstyle{wide dkpointdag}=[wide kpointdag,doubled]
\tikzstyle{dkcopoint}=[kpoint adjoint,doubled]
\tikzstyle{dkpointadj}=[kpoint adjoint,doubled]
\tikzstyle{dkpointconj}=[kpoint conjugate,doubled]
\tikzstyle{dkpointtrans}=[kpoint transpose,doubled]

\tikzstyle{kscalar}=[kpoint common, shape=EBox, inner xsep=-1pt, inner ysep=3pt,font=\small]
\tikzstyle{kscalarconj}=[kpoint common, shape=WBox, inner xsep=-1pt, inner ysep=3pt,font=\small]

\tikzstyle{spekpoint}=[kpoint sc,minimum height=5mm,inner sep=3pt]
\tikzstyle{spekcopoint}=[kpoint adjoint sc,minimum height=5mm,inner sep=3pt]

\tikzstyle{dspekpoint}=[spekpoint,doubled]
\tikzstyle{dspekcopoint}=[spekcopoint,doubled]

 \tikzstyle{upground}=[circuit ee IEC,thick,ground,rotate=90,scale=1.4]
 \tikzstyle{upgroundnormal}=[circuit ee IEC,thick,ground,rotate=90,scale=2]
 \tikzstyle{downground}=[circuit ee IEC,thick,ground,rotate=-90,scale=1.4]
 \tikzstyle{bigground}=[regular polygon,regular polygon sides=3,draw=gray,scale=0.50,inner sep=-0.5pt,minimum width=10mm,fill=gray]

\tikzstyle{arrs}=[-latex,font=\small,auto]
\tikzstyle{arrow plain}=[arrs]
\tikzstyle{arrow dashed}=[dashed,arrs]
\tikzstyle{arrow bold}=[very thick,arrs]
\tikzstyle{arrow hide}=[draw=white!0,-]
\tikzstyle{arrow reverse}=[latex-]
\tikzstyle{cdnode}=[]

%% file: figures/compound-process-time.tikz
\begin{tikzpicture}
	\begin{pgfonlayer}{nodelayer}
		\node [style=medium box] (0) at (0, 1) {$g$};
		\node [style=none] (1) at (-1.25, -0.75) {};
		\node [style=none] (2) at (-0.75, 0.5) {};
		\node [style=none] (3) at (-0.75, 2.5) {};
		\node [style=none] (4) at (-0.75, 1.5) {};
		\node [style=right label] (5) at (-2, 2.25) {$A$};
		\node [style=none] (6) at (0.75, 0.5) {};
		\node [style=none] (7) at (1.25, -0.75) {};
		\node [style=medium box] (8) at (-1.75, -1.25) {$f$};
		\node [style=none] (9) at (-2.25, -0.75) {};
		\node [style=none] (10) at (-2.25, 2.5) {};
		\node [style=none] (11) at (1.25, -2.5) {};
		\node [style=none] (12) at (1.25, -1.75) {};
		\node [style=small box] (13) at (1.25, -1.25) {$h$};
		\node [style=none] (14) at (0.75, 1.5) {};
		\node [style=none] (15) at (0.75, 2.5) {};
		\node [style=right label] (16) at (-0.5, 2.25) {$B$};
		\node [style=right label] (17) at (1, 2.25) {$C$};
		\node [style=right label] (18) at (-0.75, -0.25) {$A$};
		\node [style=right label] (19) at (1.25, 0) {$D$};
		\node [style=right label] (20) at (1.5, -2.25) {$A$};
		\node [style=none] (21) at (2.5, 1.5) {};
		\node [style=none] (22) at (2.5, -1.5) {};
		\node [style=none] (23) at (2.75, 0) {\color{gray}$t$};
	\end{pgfonlayer}
	\begin{pgfonlayer}{edgelayer}
		\draw [in=-90, out=90, looseness=1.00] (1.center) to (2.center);
		\draw (4.center) to (3.center);
		\draw [in=-90, out=90, looseness=1.00] (7.center) to (6.center);
		\draw (9.center) to (10.center);
		\draw (11.center) to (12.center);
		\draw (14.center) to (15.center);
		\draw [diredge, gray] (22.center) to (21.center);
	\end{pgfonlayer}
\end{tikzpicture}

%% file: figures/disc.tikz
\begin{tikzpicture}
	\begin{pgfonlayer}{nodelayer}
		\node [style=none] (0) at (0, -0.75) {};
		\node [style=upground] (1) at (0, 0.75) {};
	\end{pgfonlayer}
	\begin{pgfonlayer}{edgelayer}
		\draw [style=swap] (1) to (0.center); 
	\end{pgfonlayer}
\end{tikzpicture}

%% file: figures/disc-effect.tikz
\begin{tikzpicture}
	\begin{pgfonlayer}{nodelayer}
		\node [style=box] (0) at (0, -0.5) {$f$};
		\node [style=none] (1) at (0, -0.5) {};
		\node [style=none] (2) at (0, -1.75) {};
		\node [style=empty diagram] (3) at (0, 1.5) {};
	\end{pgfonlayer}
	\begin{pgfonlayer}{edgelayer}
		\draw (2.center) to (1);
	\end{pgfonlayer}
\end{tikzpicture}

%% file: figures/term.tikz
\begin{tikzpicture}
	\begin{pgfonlayer}{nodelayer}
		\node [style=box] (0) at (-1.5, 0) {$f$};
		\node [style=upground] (1) at (-1.5, 1.25) {};
		\node [style=none] (2) at (-1.5, -1.25) {};
		\node [style=none] (3) at (0.25, 0) {$=$};
		\node [style=upground] (4) at (1.5, 0.75) {};
		\node [style=none] (5) at (1.5, -0.75) {};
	\end{pgfonlayer}
	\begin{pgfonlayer}{edgelayer} 
		\draw (2.center) to (1); 
		\draw (5.center) to (4);
	\end{pgfonlayer}
\end{tikzpicture}

%% file: figures/uniquenessOfEffects.tikz
\begin{tikzpicture}
	\begin{pgfonlayer}{nodelayer}
		\node [style=box] (0) at (-1.5, -0) {\raisebox{\depth}{\scalebox{1}[1]{$f$}}};
		\node [style=none] (1) at (-1.5, -0) {};
		\node [style=none] (2) at (-1.5, -1.5) {};
		\node [style=none] (3) at (0.25, -0) {$=$};
		\node [style=none] (4) at (3.75, -0) {$=$};
		\node [style=none] (5) at (5, -1.5) {};
		\node [style=upground] (6) at (5, 0.25) {};
		\node [style=none] (7) at (2, -1.5) {};
		\node [style=box] (8) at (2, -0) {\raisebox{\depth}{\scalebox{1}[1]{$f$}}};
		\node [style=none] (9) at (2, -0) {};
		\node [style={empty diagram}] (10) at (2, 2) {};
	\end{pgfonlayer}
	\begin{pgfonlayer}{edgelayer}
		\draw (2.center) to (1.center);
		\draw (5.center) to (6);
		\draw (7.center) to (9.center);
	\end{pgfonlayer}
\end{tikzpicture}

%% file: figures/compound-process-time-rev.tikz
\begin{tikzpicture}
	\begin{pgfonlayer}{nodelayer}
		\node [style=medium box] (0) at (0, 1) {$g$};
		\node [style=none] (1) at (-1.25, -0.75) {};
		\node [style=none] (2) at (-0.75, 0.5) {};
		\node [style=none] (3) at (-0.75, 2.5) {};
		\node [style=none] (4) at (-0.75, 1.5) {};
		\node [style=right label] (5) at (-2, 2.25) {$A$};
		\node [style=none] (6) at (0.75, 0.5) {};
		\node [style=none] (7) at (1.25, -0.75) {};
		\node [style=medium box] (8) at (-1.75, -1.25) {$f$};
		\node [style=none] (9) at (-2.25, -0.75) {};
		\node [style=none] (10) at (-2.25, 2.5) {};
		\node [style=none] (11) at (1.25, -2.5) {};
		\node [style=none] (12) at (1.25, -1.75) {};
		\node [style=small box] (13) at (1.25, -1.25) {$h$};
		\node [style=none] (14) at (0.75, 1.5) {};
		\node [style=none] (15) at (0.75, 2.5) {};
		\node [style=right label] (16) at (-0.5, 2.25) {$B$};
		\node [style=right label] (17) at (1, 2.25) {$C$};
		\node [style=right label] (18) at (-0.75, -0.25) {$A$};
		\node [style=right label] (19) at (1.25, 0) {$D$};
		\node [style=right label] (20) at (1.5, -2.25) {$A$};
		\node [style=none] (21) at (2.5, -1.5) {};
		\node [style=none] (22) at (2.5, 1.5) {};
		\node [style=none] (23) at (2.75, 0) {\color{gray}$t$};
	\end{pgfonlayer}
	\begin{pgfonlayer}{edgelayer}
		\draw [in=-90, out=90, looseness=1.00] (1.center) to (2.center);
		\draw (4.center) to (3.center);
		\draw [in=-90, out=90, looseness=1.00] (7.center) to (6.center);
		\draw (9.center) to (10.center);
		\draw (11.center) to (12.center);
		\draw (14.center) to (15.center);
		\draw [diredge, gray] (22.center) to (21.center);
	\end{pgfonlayer}
\end{tikzpicture}

%% file: figures/compound-process-time-flip.tikz
\begin{tikzpicture}
	\begin{pgfonlayer}{nodelayer}
		\node [style=medium box] (0) at (0, -1) {\raisebox{\depth}{\scalebox{1}[-1]{$g$}}};
		\node [style=none] (1) at (-1.25, 0.75) {};
		\node [style=none] (2) at (-0.75, -0.5) {};
		\node [style=none] (3) at (-0.75, -2.5) {};
		\node [style=none] (4) at (-0.75, -1.5) {};
		\node [style=right label] (5) at (-2, -2.25) {$A$};
		\node [style=none] (6) at (0.75, -0.5) {};
		\node [style=none] (7) at (1.25, 0.75) {};
		\node [style=medium box] (8) at (-1.75, 1.25) {\raisebox{\depth}{\scalebox{1}[-1]{$f$}}};
		\node [style=none] (9) at (-2.25, 0.75) {};
		\node [style=none] (10) at (-2.25, -2.5) {};
		\node [style=none] (11) at (1.25, 2.5) {};
		\node [style=none] (12) at (1.25, 1.75) {};
		\node [style=small box] (13) at (1.25, 1.25) {\raisebox{\depth}{\scalebox{1}[-1]{$h$}}};
		\node [style=none] (14) at (0.75, -1.5) {};
		\node [style=none] (15) at (0.75, -2.5) {};
		\node [style=right label] (16) at (-0.5, -2.25) {$B$};
		\node [style=right label] (17) at (1, -2.25) {$C$};
		\node [style=right label] (18) at (-0.75, 0.25) {$A$};
		\node [style=right label] (19) at (1.25, 0) {$D$};
		\node [style=right label] (20) at (1.5, 2.25) {$A$};
		\node [style=none] (21) at (2.5, 1.5) {};
		\node [style=none] (22) at (2.5, -1.5) {};
		\node [style=none] (23) at (2.75, 0) {\color{gray}$t$};
	\end{pgfonlayer}
	\begin{pgfonlayer}{edgelayer}
		\draw [in=90, out=-90, looseness=1.00] (1.center) to (2.center);
		\draw (4.center) to (3.center);
		\draw [in=90, out=-90, looseness=1.00] (7.center) to (6.center);
		\draw (9.center) to (10.center);
		\draw (11.center) to (12.center);
		\draw (14.center) to (15.center);
		\draw [diredge, gray] (22.center) to (21.center);
	\end{pgfonlayer}
\end{tikzpicture}

%% file: figures/disc-rev.tikz
\begin{tikzpicture}
	\begin{pgfonlayer}{nodelayer}
		\node [style=none] (0) at (0, 0.75) {};
		\node [style=downground] (1) at (0, -0.75) {};
	\end{pgfonlayer}
	\begin{pgfonlayer}{edgelayer}
		\draw [style=swap] (1) to (0.center);
	\end{pgfonlayer}
\end{tikzpicture}

%% file: figures/term-rev.tikz
\begin{tikzpicture}
	\begin{pgfonlayer}{nodelayer}
		\node [style=box] (0) at (-1.5, 0) {\raisebox{\depth}{\scalebox{1}[-1]{$f$}}};
		\node [style=downground] (1) at (-1.5, -1.25) {};
		\node [style=none] (2) at (-1.5, 1.25) {};
		\node [style=none] (3) at (0.25, 0) {$=$};
		\node [style=downground] (4) at (1.5, -0.75) {};
		\node [style=none] (5) at (1.5, 0.75) {};
	\end{pgfonlayer}
	\begin{pgfonlayer}{edgelayer}
		\draw (2.center) to (1); 
		\draw (5.center) to (4); 
	\end{pgfonlayer}
\end{tikzpicture}

%% file: figures/unitary.tikz
\begin{tikzpicture}
	\begin{pgfonlayer}{nodelayer}
		\node [style=box] (0) at (-1.5, 0.75) {\raisebox{\depth}{\scalebox{1}[-1]{$U$}}};
		\node [style=none] (1) at (-1.5, 0.75) {};
		\node [style=none] (2) at (-1.5, 1.75) {};
		\node [style=none] (3) at (0.25, 0) {$=$};
		\node [style=box] (4) at (-1.5, -0.75) {\raisebox{\depth}{\scalebox{1}[1]{$U$}}};
		\node [style=none] (5) at (-1.5, -1.75) {};
		\node [style=none] (6) at (1.5, 1) {};
		\node [style=none] (7) at (1.5, -1) {};
	\end{pgfonlayer}
	\begin{pgfonlayer}{edgelayer}
		\draw (2.center) to (1.center);
		\draw (0) to (4);
		\draw (4) to (5.center);
		\draw (6.center) to (7.center);
	\end{pgfonlayer}
\end{tikzpicture}

%% file: figures/meas-explicit.tikz
\begin{tikzpicture}
	\begin{pgfonlayer}{nodelayer}
		\node [style=none] (0) at (-0.75, 1.25) {};
		\node [style=none] (1) at (-0.75, -1.25) {};
		\node [style=none] (2) at (0.75, 1.25) {};
		\node [style=white dot] (3) at (-0.75, 0) {};
		\node [style=none] (4) at (-0.75, 0) {};
		\node [style=none] (5) at (1.25,0) {$=$};
		\node [style=none] (6) at (2.75,-0.25) {$\sum\limits_{x \in X}$};
		\node [style=box] (7) at (4.5,0) {$P_x$};
		\node [style=none] (8) at (4.5, 1.25) {};
		\node [style=none] (9) at (4.5, -1.25) {};
		\node [style=point] (10) at (6,0) {$x$};
		\node [style=none] (11) at (6, 1.25) {};
	\end{pgfonlayer}
	\begin{pgfonlayer}{edgelayer}
		\draw [style=boldedge] (0.center) to (1.center);
		\draw [style=swap, in=-90, out=15, looseness=1.00] (4.center) to (2.center);
		\draw [style=boldedge] (8.center) to (9.center);
		\draw [-] (10.center) to (11.center);
	\end{pgfonlayer}
\end{tikzpicture}

%% file: figures/dem-meas-explicit.tikz
\begin{tikzpicture}
	\begin{pgfonlayer}{nodelayer}
		\node [style=upground] (0) at (-0.75, 1.25) {};
		\node [style=none] (1) at (-0.75, -1.25) {};
		\node [style=none] (2) at (0.75, 1.25) {};
		\node [style=white dot] (3) at (-0.75, 0) {};
		\node [style=none] (4) at (-0.75, 0) {};
		\node [style=none] (5) at (1.25,0) {$=$};
		\node [style=none] (6) at (2.75,-0.25) {$\sum\limits_{x \in X}$};
		\node [style=box] (7) at (4.5,0) {$P_x$};
		\node [style=upground] (8) at (4.5, 1.25) {};
		\node [style=none] (9) at (4.5, -1.25) {};
		\node [style=point] (10) at (6,0) {$x$};
		\node [style=none] (11) at (6, 1.25) {};
		\node [style=none] (12) at (-3.25, -1.25) {};
		\node [style=none] (13) at (-3.25, 1.25) {};
		\node [style=white dot] (14) at (-3.25, 0) {};
		\node [style=none] (15) at (-2,0) {$:=$};
	\end{pgfonlayer}
	\begin{pgfonlayer}{edgelayer}
		\draw [style=boldedge] (0) to (1.center);
		\draw [style=swap, in=-90, out=15, looseness=1.00] (4.center) to (2.center);
		\draw [style=boldedge] (8) to (9.center);
		\draw [-] (10.center) to (11.center);
		\draw [style=boldedge] (14) to (12.center);
		\draw [-] (14) to (13.center);
	\end{pgfonlayer}
\end{tikzpicture}

%% file: figures/meas-rev-explicit.tikz
\begin{tikzpicture}
	\begin{pgfonlayer}{nodelayer}
		\node [style=none] (0) at (-0.75, 1.25) {};
		\node [style=none] (1) at (-0.75, -1.25) {};
		\node [style=none] (2) at (0.75, -1.25) {};
		\node [style=white dot] (3) at (-0.75, 0) {};
		\node [style=none] (4) at (-0.75, 0) {};
		\node [style=none] (5) at (1.25,0) {$=$};
		\node [style=none] (6) at (2.75,-0.25) {$\sum\limits_{x \in X}$};
		\node [style=box] (7) at (4.5,0) {$P_x$};
		\node [style=none] (8) at (4.5, 1.25) {};
		\node [style=none] (9) at (4.5, -1.25) {};
		\node [style=copoint] (10) at (6,0) {$x$};
		\node [style=none] (11) at (6, -1.25) {};
	\end{pgfonlayer}
	\begin{pgfonlayer}{edgelayer}
		\draw [style=boldedge] (0.center) to (1.center);
		\draw [style=swap, in=90, out=-15, looseness=1.00] (4.center) to (2.center);
		\draw [style=boldedge] (8.center) to (9.center);
		\draw [-] (10.center) to (11.center);
	\end{pgfonlayer}
\end{tikzpicture}

%% file: figures/dem-meas-rev-explicit.tikz
\begin{tikzpicture}
	\begin{pgfonlayer}{nodelayer}
		\node [style=none] (0) at (-0.75, 1.25) {};
		\node [style=downground] (1) at (-0.75, -1.25) {};
		\node [style=none] (2) at (0.75, -1.25) {};
		\node [style=white dot] (3) at (-0.75, 0) {};
		\node [style=none] (4) at (-0.75, 0) {};
		\node [style=none] (5) at (1.25,0) {$=$};
		\node [style=none] (6) at (2.75,-0.25) {$\sum\limits_{x \in X}$};
		\node [style=box] (7) at (4.5,0) {$P_x$};
		\node [style=none] (8) at (4.5, 1.25) {};
		\node [style=downground] (9) at (4.5, -1.25) {};
		\node [style=copoint] (10) at (6,0) {$x$};
		\node [style=none] (11) at (6, -1.25) {};
		\node [style=none] (12) at (-3.25, -1.25) {};
		\node [style=none] (13) at (-3.25, 1.25) {};
		\node [style=white dot] (14) at (-3.25, 0) {};
		\node [style=none] (15) at (-2,0) {$:=$};
	\end{pgfonlayer}
	\begin{pgfonlayer}{edgelayer}
		\draw [style=boldedge] (0.center) to (1);
		\draw [style=swap, in=90, out=-15, looseness=1.00] (4.center) to (2.center);
		\draw [style=boldedge] (8.center) to (9);
		\draw [-] (10.center) to (11.center);
		\draw [-] (14) to (12.center);
		\draw [style=boldedge] (14) to (13.center);
	\end{pgfonlayer}
\end{tikzpicture}

%% file: figures/meas-rev-explicit-uniform.tikz
\begin{tikzpicture}
	\begin{pgfonlayer}{nodelayer}
		\node [style=none] (0) at (-0.75, 1.25) {};
		\node [style=none] (1) at (-0.75, -1.25) {};
		\node [style=downground] (2) at (0.75, -1.25) {};
		\node [style=white dot] (3) at (-0.75, 0) {};
		\node [style=none] (4) at (-0.75, 0) {};
		\node [style=none] (5) at (1.25,0) {$=$};
		\node [style=none] (6) at (2.75,-0.25) {$\sum\limits_{x \in X}$};
		\node [style=box] (7) at (4.5,0) {$P_x$};
		\node [style=none] (8) at (4.5, 1.25) {};
		\node [style=none] (9) at (4.5, -1.25) {};
	\end{pgfonlayer}
	\begin{pgfonlayer}{edgelayer}
		\draw [style=boldedge] (0.center) to (1.center);
		\draw [style=swap, in=90, out=-15, looseness=1.00] (4.center) to (2.180);
		\draw [style=boldedge] (8.center) to (9.center);
	\end{pgfonlayer}
\end{tikzpicture}

%% file: figures/dem-meas-rev-explicit-uniform.tikz
\begin{tikzpicture}
	\begin{pgfonlayer}{nodelayer}
		\node [style=none] (0) at (-0.75, 1.25) {};
		\node [style=downground] (1) at (-0.75, -1.25) {};
		\node [style=downground] (12) at (-3.25, -1.25) {};
		\node [style=none] (13) at (-3.25, 1.25) {};
		\node [style=white dot] (14) at (-3.25, 0) {};
		\node [style=none] (15) at (-2,0) {$=$};
	\end{pgfonlayer}
	\begin{pgfonlayer}{edgelayer}
		\draw [style=boldedge] (0.center) to (1);
		\draw [-] (14) to (12);
		\draw [style=boldedge] (14) to (13.center);
	\end{pgfonlayer}
\end{tikzpicture}

%% file: figures/measdiag4.tikz
\begin{tikzpicture}
	\begin{pgfonlayer}{nodelayer}
		\node [style=none] (0) at (0.75, -0.75) {};
		\node [style=none] (1) at (-0.75, -0.75) {};
		\node [style=none] (2) at (0.75, 0.75) {};
		\node [style=none] (3) at (-0.75, 0.75) {};
		\node [style=none] (4) at (-1, -0.75) {};
		\node [style=none] (5) at (-1, 0.75) {};
		\node [style=none] (6) at (1, 0.75) {};
		\node [style=none] (7) at (1, -0.75) {};
		\node [style=none] (8) at (-0.75, -1) {};
		\node [style=none] (9) at (-0.75, -0.75) {};
		\node [style=none] (10) at (0.75, -1) {};
		\node [style=none] (11) at (0.75, -0.75) {};
		\node [style=none] (12) at (-0.75, 0.75) {}; 
		\node [style=none] (13) at (-0.75, 1) {};
		\node [style=none] (14) at (0.75, 0.75) {};
		\node [style=none] (15) at (0.75, 1) {};
	\end{pgfonlayer}
	\begin{pgfonlayer}{edgelayer}
		\draw [style=boldedge, bend left=90, looseness=1.25] (1.center) to (0.center);
		\draw [style=boldedge, bend right=90, looseness=1.25] (3.center) to (2.center);
		\draw [style=dashed] (5.center) to (4.center);
		\draw [style=dashed] (4.center) to (7.center);
		\draw [style=dashed] (6.center) to (7.center);
		\draw [style=dashed] (5.center) to (6.center);
		\draw [style=boldedge] (9.center) to (8.center);
		\draw [style=boldedge] (11.center) to (10.center);
		\draw [style=boldedge] (13.center) to (12.center);
		\draw [style=boldedge] (15.center) to (14.center);
	\end{pgfonlayer}
\end{tikzpicture}

%% file: figures/measdiag7.tikz
\begin{tikzpicture}
	\begin{pgfonlayer}{nodelayer}
		\node [style=none] (0) at (0.75, -0.75) {};
		\node [style=none] (1) at (-0.75, -0.75) {};
		\node [style=none] (2) at (0.75, 0.75) {};
		\node [style=none] (3) at (-0.75, 0.75) {};
		\node [style=none] (4) at (-3, -2) {};
		\node [style=none] (5) at (-3, 2) {};
		\node [style=none] (6) at (3.25, 2) {};
		\node [style=none] (7) at (3.25, -2) {};
		\node [style=none] (8) at (-0.75, -2.25) {};
		\node [style=none] (9) at (-0.75, -0.75) {};
		\node [style=none] (10) at (0.75, -1) {};
		\node [style=none] (11) at (0.75, -0.75) {};
		\node [style=none] (12) at (-0.75, 0.75) {};
		\node [style=none] (13) at (-0.75, 2.25) {};
		\node [style=none] (14) at (0.75, 0.75) {};
		\node [style=none] (15) at (0.75, 1) {};
		\node [style=none] (16) at (-2, -0.25) {$\displaystyle\sum_x$};
		\node [style=box] (17) at (0.75, 1.25) {$U_x$};
		\node [style=box] (18) at (0.75, -1.25) {$U_x^\dagger$};
		\node [style=copoint] (19) at (2.25, 0) {$x$};
		\node [style=point] (20) at (2.25, -2.75) {$0$};
		\node [style=none] (21) at (0.75, 2.25) {};
		\node [style=none] (22) at (0.75, 1.5) {};
		\node [style=none] (23) at (0.75, 1.5) {};
		\node [style=none] (24) at (0.75, -1.5) {};
		\node [style=none] (25) at (0.75, -2.25) {};
		\node [style=left label] (26) at (0, -3) {\color{gray}\footnotesize `pre-selectionâ};
		\node [style=none] (27) at (1.75, -2.75) {};
		\node [style=none] (28) at (0.25, -3) {};
	\end{pgfonlayer}
	\begin{pgfonlayer}{edgelayer}
		\draw [style=boldedge, bend left=90, looseness=1.25] (1.center) to (0.center);
		\draw [style=boldedge, bend right=90, looseness=1.25] (3.center) to (2.center);
		\draw [style=dashed] (5.center) to (4.center);
		\draw [style=dashed] (4.center) to (7.center);
		\draw [style=dashed] (6.center) to (7.center);
		\draw [style=dashed] (5.center) to (6.center);
		\draw [style=boldedge] (9.center) to (8.center);
		\draw [style=boldedge] (11.center) to (10.center);
		\draw [style=boldedge] (13.center) to (12.center);
		\draw [style=boldedge] (15.center) to (14.center);
		\draw (19) to (20);
		\draw [style=boldedge] (21.center) to (22.center);
		\draw [style=boldedge] (24.center) to (25.center);
		\draw [style=pointer edge, in=180, out=0, looseness=0.75] (28.center) to (27.center);
	\end{pgfonlayer}
\end{tikzpicture}

%% file: figures/bell-teleportation.tikz
\begin{tikzpicture}
	\begin{pgfonlayer}{nodelayer}
		\node [style=none] (0) at (0, -0.5) {};
		\node [style=none] (1) at (-1.5, -0.5) {};
		\node [style=none] (2) at (1.5, -0.5) {};
		\node [style=none] (3) at (0, -0.5) {};
		\node [style=none] (4) at (-1.5, -1.5) {};
		\node [style=none] (5) at (-1.5, -0.5) {};
		\node [style=none] (6) at (0, -0.5) {};
		\node [style=none] (7) at (0, -0.5) {};
		\node [style=none] (8) at (1.5, -0.5) {};
		\node [style=none] (9) at (1.5, 1.25) {};
		\node [style=none] (10) at (2.5, 0) {$=$};
		\node [style=none] (11) at (0, 1.25) {};
		\node [style=none] (12) at (-1.5, 1.25) {};
		\node [style=none] (13) at (-1.75, -0.25) {};
		\node [style=none] (14) at (-1.75, 1) {};
		\node [style=none] (15) at (0.25, 1) {};
		\node [style=none] (16) at (0.25, -0.25) {};
		\node [style=none] (17) at (3.5, -1.5) {};
		\node [style=none] (18) at (6.5, 1.25) {};
		\node [style=none] (19) at (5, 1.25) {};
		\node [style=none] (20) at (3.5, 1.25) {};
	\end{pgfonlayer}
	\begin{pgfonlayer}{edgelayer}
		\draw [style=boldedge, bend left=90, looseness=1.75] (1.center) to (0.center);
		\draw [style=boldedge, bend right=90, looseness=1.75] (3.center) to (2.center);
		\draw [style=boldedge] (5.center) to (4.center);
		\draw [style=boldedge] (9.center) to (8.center);
		\draw [style=boldedge, bend left=90, looseness=1.75] (11.center) to (12.center);
		\draw [style=dashed] (14.center) to (13.center);
		\draw [style=dashed] (13.center) to (16.center);
		\draw [style=dashed] (15.center) to (16.center);
		\draw [style=dashed] (14.center) to (15.center);
		\draw [style=boldedge, in=-90, out=90, looseness=1.00] (17.center) to (18.center);
		\draw [style=boldedge, bend left=90, looseness=1.75] (19.center) to (20.center);
	\end{pgfonlayer}
\end{tikzpicture}

%% file: figures/comp-basis-prep.tikz
\begin{tikzpicture}
	\begin{pgfonlayer}{nodelayer}
		\node [style=point] (0) at (-1.5, -1) {$x$};
		\node [style=none] (1) at (0, 0) {$=$};
		\node [style=dpoint] (2) at (1.75, -0.5) {$x$};
		\node [style=none] (3) at (1.75, -0.25) {};
		\node [style=none] (4) at (1.75, 0.75) {};
		\node [style=none] (5) at (-1.5, 0.25) {};
		\node [style=white dot] (6) at (-1.5, 0.25) {};
		\node [style=none] (7) at (-1.5, -1) {};
		\node [style=none] (8) at (-1.5, 1.25) {};
		\node [style=none] (9) at (-1.5, 0.25) {};
		\node [style=none] (10) at (6,0) {for $x \in \{0,1\}$};
	\end{pgfonlayer}
	\begin{pgfonlayer}{edgelayer}
		\draw [style=boldedge] (4.center) to (3.center);
		\draw [style=swap, in=-90, out=90, looseness=1.00] (7.center) to (5.center);
		\draw [style=boldedge] (8.center) to (9.center);
	\end{pgfonlayer}
\end{tikzpicture}